\documentclass[runningheads]{llncs}

\usepackage[pdftex]{graphicx}
	\pdfoutput=1

\usepackage{color}
	\definecolor{color1}{rgb}{0,0,0.8}

\usepackage[colorlinks=true,
	linkcolor=black,
	citecolor=color1,
	filecolor=color1,
	menucolor=color1,
	urlcolor=color1,
	pdfpagelayout=TwoColumnRight,
	pdftitle={How to Extract the Geometry and Topology from Very Large 3D Segmentations},
	pdfsubject={Computational Geometry, Segmentation},
	pdfauthor={Bjoern Andres, Ullrich Koethe, Thorben Kroeger and Fred A. Hamprecht},
	pdfkeywords={Computational Geometry,Digital Geometry,Segmentation,Cellular Complex,Cell Complex,CW Complex,Topological Grid}]{hyperref}

\usepackage{amsmath,amssymb}

\usepackage[ruled,linesnumbered]{algorithm2e}
  \SetKw{Break}{break}
  \SetFuncSty{textrm}
  \SetArgSty{textrm}
  \SetCommentSty{textit}
  \SetKwBlock{Forever}{repeat}{end}
  \SetKw{KwSide}{Side effects: }

\usepackage{array}

\usepackage{sidecap}

\newcommand{\bluf}[1]{\color{blue}{\bf{#1}}}
\newcommand{\gref}[1]{\color{green}{\bf{#1}}}
\newcommand{\magf}[1]{\color{magenta}{\bf{#1}}}
\newcommand{\cyf}[1]{\color{cyan}{\bf{#1}}}

\begin{document}
\pagestyle{headings}

\mainmatter
\title{How to Extract the Geometry and Topology from Very Large 3D Segmentations}
\titlerunning{Geometry and Topology Extraction from Very Large 3D Segmentations}
\authorrunning{B.~Andres et al.}
\author{Bjoern Andres, Ullrich Koethe, Thorben Kroeger and Fred A.~Hamprecht}
\institute{HCI, IWR, University of Heidelberg\\
	\href{http://hci.iwr.uni-heidelberg.de}{http://hci.iwr.uni-heidelberg.de},
	\href{mailto:bjoern.andres@iwr.uni-heidelberg.de}{bjoern.andres@iwr.uni-heidelberg.de}
}

\maketitle

\begin{abstract}
Segmentation is often an essential intermediate step in image analysis. A volume
segmentation characterizes the underlying volume image in terms of geometric
information--segments, faces between segments, curves in which several faces
meet--as well as a topology on these objects. Existing algorithms encode this
information in designated data structures, but require that these data
structures fit entirely in Random Access Memory (RAM). Today, 3D images with 
several billion voxels are acquired, e.g. in structural neurobiology. Since 
these large volumes can no longer be processed with existing methods, we present 
a new algorithm which performs geometry and topology extraction with a runtime 
linear in the number of voxels and log-linear in the number of faces and curves. 
The parallelizable algorithm proceeds in a block-wise fashion and constructs a
consistent representation of the entire volume image on the hard drive, making
the structure of very large volume segmentations accessible to image analysis.
The parallelized ``CGP'' C++ source code, free command line tools and MATLAB mex
files are avilable from \url{http://hci.iwr.uni-heidelberg.de/software.php}.
\end{abstract}

\section{Introduction}

Segmentations of volume images partition the volume into different connected
components: \emph{segments}, \emph{faces} between segments, \emph{curves} in 
which several faces meet, as well as the \emph{points} between these curves, 
(Fig. \ref{figure:geo-boundaries-lines}). The geometry and topology of these 
components are essential in many analyses. In order to compute features that
describe this geometry and topology, a data structure is needed that provides
fast access to all components and their adjacency. Volume segmentations are 
usually stored simply as \emph{volume labelings}, i.e.~as 3-dimensional arrays
in which each entry is a label that uniquely identifies the segment to which 
the voxel belongs. This form of storage does not represent geometry and 
topology explicitly. Instead, faces between segments and the curves between 
these faces are encoded only implicitly, as adjacent voxels whose segment 
labels differ. All that can be obtained from an array in constant computation 
time is the segment label at a given voxel. Neither is the set of voxels that 
belong to the same segment readily available, nor are the faces between 
adjacent segments or the curves in which several of these faces meet. It is not
stored explicitly which segments are adjacent, separated by which faces, and in 
which curves adjacent faces meet.

The new algorithm presented in this article takes a volume labeling as input 
and extracts the geometry and topology of all components in a block-wise 
fashion, in a runtime that is linear in the number of voxels and log-linear in 
the number of faces and curves. Blocks of the volume labeling can be processed 
either sequentially, on a single computer that might have only a few hundred 
megabytes of RAM, or in parallel, on several computers, which facilitates 
geometry and topology extraction from datasets that consist of more than $10^9$
voxels. In both cases, a consistent representation of the entire volume 
segmentation is constructed on the hard drive, in a data structure from which 
all connected components and their adjacency can be obtained in constant 
computation time. The new algorithm makes the geometry and topology of large 
volume segmentations accessible to image analysis. Its parallelized C++ source 
code, free command line tools and MATLAB mex files are provided.

This article is organized as follows: Related work is discussed in 
Section~\ref{section:geo-related-work}. 
In 
Section~\ref{section:geo-representation}, the data structure that captures the 
geometry and topology of a volume segmentation is introduced. The algorithm for 
its construction is defined in 
Section~\ref{section:geo-algorithm} 
and extended in 
Section~\ref{section:geo-block-wise} 
to work with limited RAM and in parallel. The correctness of the algorithm is 
proved and its complexity analyzed. 
Section~\ref{section:geo-redundant-storage} 
describes the efficient storage of the data structure on the hard drive, and 
Section~\ref{section:geo-conclusion} 
concludes the article.
\vfill

\begin{figure*}
\centering
\includegraphics[width=0.49\textwidth]{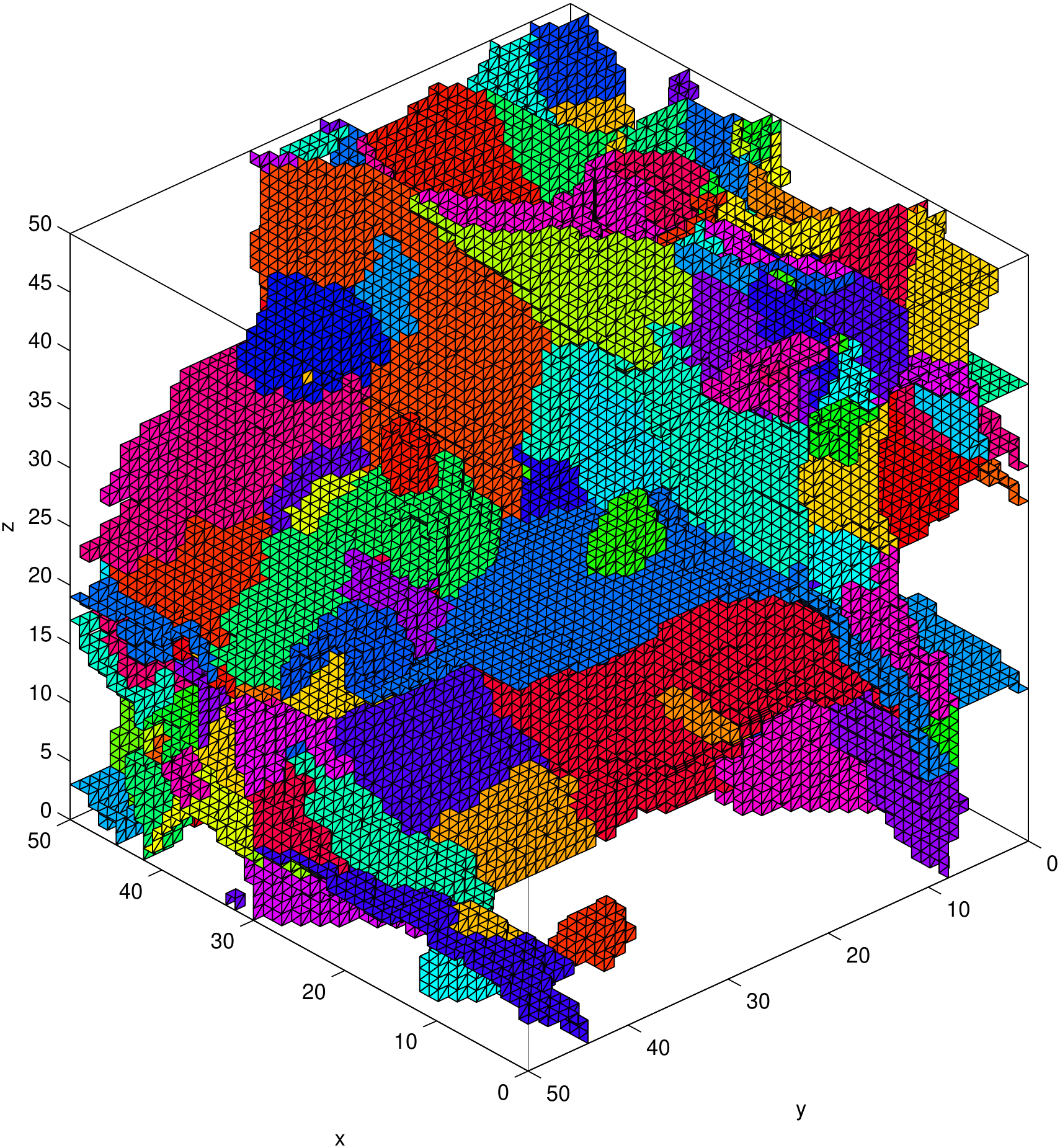}
\includegraphics[width=0.49\textwidth]{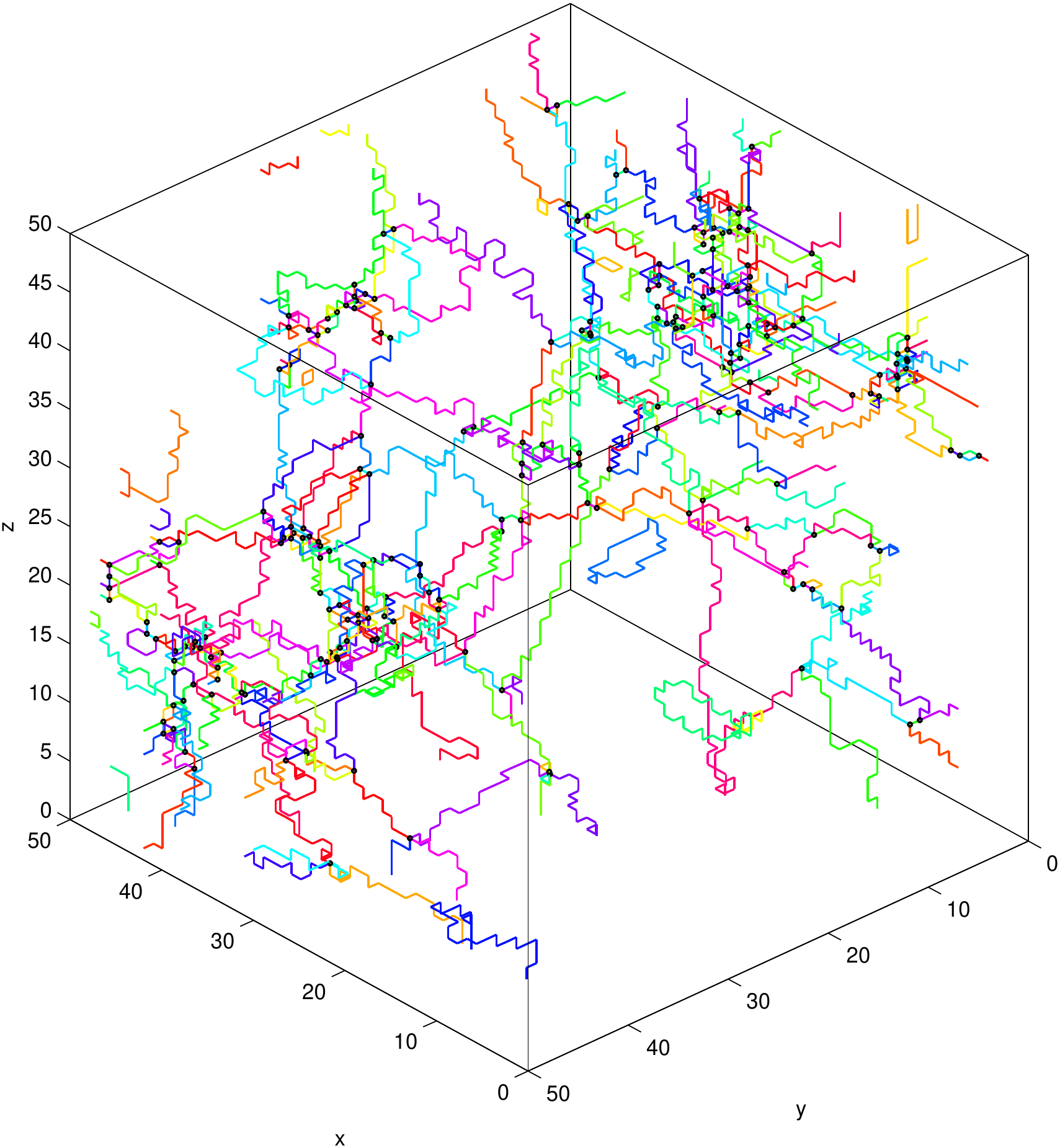}
\caption{A volume segmentation consists of segments, faces between
adjacent segments (left), the curves in which several of these faces meet,
as well as the points between these curves (right).}
\label{figure:geo-boundaries-lines}
\end{figure*}

\section{Related Work}
\label{section:geo-related-work}
The first explicit representation of all components of an image segmentation 
was proposed by Brice and Fennema
\cite{brice-1970}; 
it encodes segments as sets of pixels, curves between segments as sets of 
inter-pixel edges, and the end points of these curves as pixel corners. Naive 
attempts to represent the different components of a segmentation all as sets of
pixels on the pixel grid of the underlying image are topologically 
inconsistent, as was shown in 
\cite{pavlidis-1977} 
and proven generally and rigorously in
\cite{kovalevsky-1989,kovalevsky-1993}. 
To overcome this inconsistency, Khalimsky
\cite{khalimsky-1990} 
introduced the \emph{topological grid} whose points correspond to pixels, 
inter-pixel edges, and pixel corners. The concept of a 3-dimensional 
topological grid is depicted in 
Fig.~\ref{figure:geo-topological-grid}. 

Data structures that store, for each component of a segmentation, all points 
on the topological grid that constitute this component were proposed and 
implemented by 
\cite{meine-2005} 
for image segmentations and envisioned by 
\cite{damiand-2008} 
for volume image segmentations. However, a storage concept that is suitable 
for large volume segmentations has so far been missing.

Along with representations of the geometry of segmentations, i.e.~their 
components, at least three different structures have been used to encode the
neighborhood system on these objects: \emph{Region Adjacency Graphs} 
(RAGs) 
\cite{pavlidis-1977} 
encode the adjacency of segments. RAGs do not capture the topology of a 
segmentation completely because several disconnected faces that separate the 
same two segments correspond to the same edge in the 
RAG. 
Kropatsch \cite{kropatsch-1995} 
introduces multiple edges and self-loops in the RAG which results in a 
multi-graph whose dual graph represents faces as vertices and the adjacency of 
faces as edges. This concept can be implemented as a data structure. However, 
both the graph and its dual have to be stored and maintained which is 
algorithmically challenging. 

\emph{Combinatorial maps} were introduced in image analysis in
\cite{braquelaire-1991} 
and are used as data structures, e.g.~in
\cite{meine-2005} 
as well as in some algorithms of the Computational Geometry Algorithms Library 
(CGAL)\footnote{\url{www.cgal.org}}. 
The extension of combinatorial maps to higher dimensions is involved but 
possible 
\cite{lienhardt-1989,lienhardt-1991} 
and has facilitated the development of the 3-dimensional topological map \cite{bertrand-2000,damiand-2003}. 
This map captures not only the topology of a segmentation but also its 
embedding into the segmented space, i.e.~containment relations and orders of 
objects
\cite{lienhardt-1989,damiand-2008}. 
It is therefore more expensive to construct and manipulate than a data 
structure that encodes only the topology.

A simple structure that encodes only the topology is a finite \emph{cellular 
complex},
cf.~\cite{munkres-1995,hatcher-2002} 
also known as a cell complex or CW-complex 
\cite{klette-2000} 
where CW stands for the two axioms closure-finiteness and weak topology,
cf.~\cite{hatcher-2002}. 
Cellular complexes were first used in image processing in
\cite{kovalevsky-1989,kovalevsky-1993}. 
Their generalization to 3D is simple and intuitive. 

The main focus of previous efforts to extract and encode the geometry and 
topology of segmentations has not been on large volume segmentations but on the
efficient processing of the merging and splitting of segments. These operations 
are required within the context of inter-active segmentation. In
\cite{meine-2004,damiand-2008},
representations of the geometry and topology are constructed incrementally, 
using random access to already constructed parts of the data structure. In 
order for these algorithms to work efficiently, the underlying data structures 
need to be kept entirely in RAM. To extract the geometry and topology of a 
volume segmentation of $2,000^3$ voxels, $2,000^3 \cdot 2^3 \cdot 4$ bytes 
$\approx 238$~GB of RAM are required for the labeling of the topological grid, 
an amount that is not available on present day desktop computers. Beyond 
$3,500^3$ voxels, even the $1$~TB of RAM of a large server are insufficient. 
The method presented in this article overcomes this limitation by means of 
block-wise processing. It makes geometry and topology extraction from large 
volume segmentations possible.

\begin{SCfigure}[3.5][t]
\centering
\includegraphics[width=3cm]{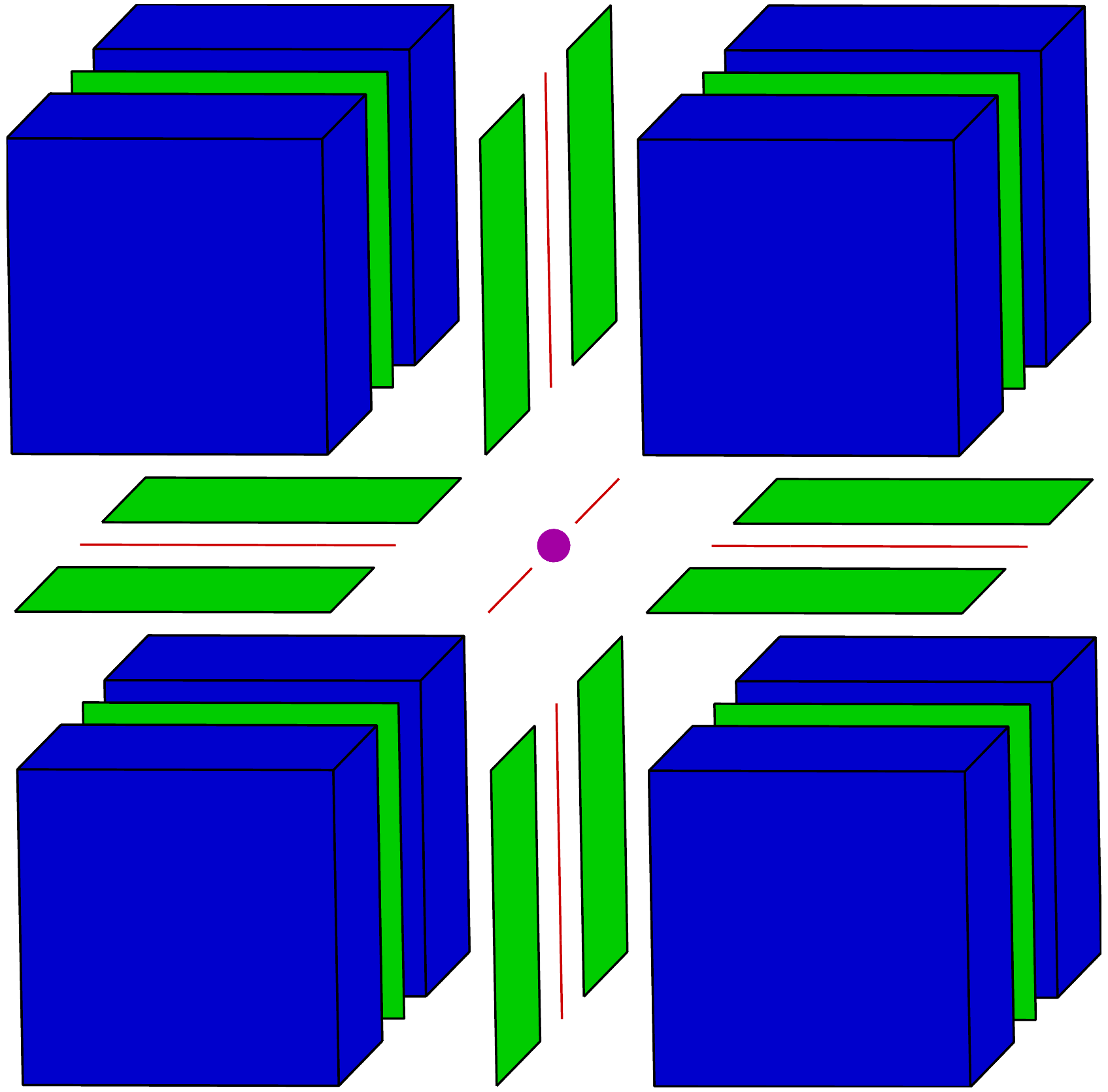}
\caption{A topological grid $T$ is used to represent the geometry of a volume 
segmentation explicitly. Its elements are called cells. A cell 
$(t_1, t_2, t_3) \in T$ with $j$ odd entries is called a $j$-cell. 3-cells, 
2-cells, 1-cells, and 0-cells respectively represent voxels (blue), faces 
between voxels (green), lines between faces (red), and points between lines 
(purple).}
\label{figure:geo-topological-grid}
\end{SCfigure}

\section{From Voxels to Geometry and Topology}
\label{section:geo-representation}
The starting point for geometry and topology extraction is a volume image on a
\emph{voxel grid} $G = \{1, \ldots, n_1\} \times \{1, \ldots, n_2\} \times 
\{1, \ldots, n_3\}$ whose extent in each of the three dimensions is given by 
$n_1, n_2, n_3 \in \mathbb{N}$. Two voxels $v,w \in G$ are said to be 
\emph{connected} if $\sum_{j=1}^{3}{|v_j - w_j|} = 1$. Each voxel is thus 
connected to 6 other voxels unless it is at the boundary of the grid. For every
voxel, the connected voxels are called its \emph{6-neighbors}. A set of voxels
$U \subseteq G$ is called connected if and only if any two distinct voxels 
$v,w \in U$ are linked by a path in $U$, i.e.~by a sequence of voxels in $U$ 
that starts with $v$ and ends with $w$, in which each voxel is connected to its 
predecessor.

A volume segmentation partitions the voxel grid $G$ into connected components 
called segments. A \emph{volume labeling}
\begin{equation}
\sigma: G \rightarrow \mathbb{N}
\end{equation}
assigns to each voxel a label that identifies the segment to which the voxel 
belongs. Since each voxel belongs to a segment, the faces between segments, the
curves between these faces and the points between these curves
(Fig.~\ref{figure:geo-boundaries-lines}) cannot be represented on the voxel 
grid. The structure that is used for this purpose is a \emph{topological grid},
\begin{equation}
T = \{1, \ldots, 2 n_1-1\} \times \{1, \ldots, 2 n_2-1\} \times \{1, \ldots, 2 n_3-1\}
\enspace .
\end{equation}
This grid has about eight times the size of the voxel grid. Its elements are 
called \emph{cells}. Cells with $j$ odd entries are called \emph{$j$-cells},
cf.~Fig.~\ref{figure:geo-topological-grid}. 

Segments, faces between segments, curves between faces, and points between 
curves correspond to connected components of cells. Two relations are crucial 
to the definition of these connected components. The first relation the 
$\Gamma$-neighborhood of cells. It is depicted in the third column of 
Fig.~\ref{figure:geo-cell_complex}.

\begin{definition}
The \emph{$\Gamma$-neighborhood} is the mapping $\Gamma: T \rightarrow 
\mathcal{P}(T)$ such that, for each $j \in \{0, \ldots, 3\}$ and any $j$-cell
$t \in T$, $\Gamma(t)$ consists of all 6-neighbors of $t$ on the
topological grid $T$ that are $(j+1)$-cells.
\end{definition}
Any 2-cell, for instance, has two $\Gamma$-neighbors that correspond to two 
voxels.

The second important relation is the \emph{connectivity} of cells; it is
depicted in the last column of Fig. \ref{figure:geo-cell_complex}.
\begin{definition}
The connectivity relation ``$\leftrightarrow$''$\ \subseteq T \times T$
connects any two cells $t_1, t_2 \in T$ (denoted $t_1 \leftrightarrow t_2$)
if and only if there exists a $t \in T$ such that both $t_1$ and $t_2$ are
$\Gamma$-neighbors of $t$.
\end{definition}

\begin{SCfigure}[1][t]
\renewcommand{\arraystretch}{1.2}
\vspace{5pt}
\centering
\begin{tabular}{m{0.021\columnwidth}m{0.138\columnwidth}m{0.138\columnwidth}m{0.138\columnwidth}}
\hline
	$n$
&	\centering $n$-cell
&	\centering $\Gamma$
&	\hspace{4mm} $\leftrightarrow$\\
\hline
	\vspace{1mm}0
&	\vspace{1mm}\centering \includegraphics[scale=0.07]{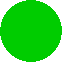}
&	\vspace{1mm}\centering \includegraphics[scale=0.07]{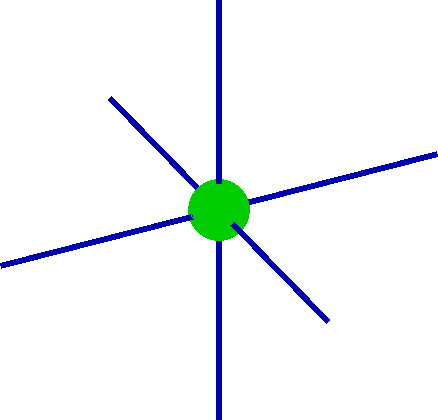}
&	\vspace{1mm}\hspace{5mm} $\emptyset$\\

	1
&	\centering \includegraphics[scale=0.07]{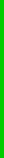}
&	\centering \includegraphics[scale=0.07]{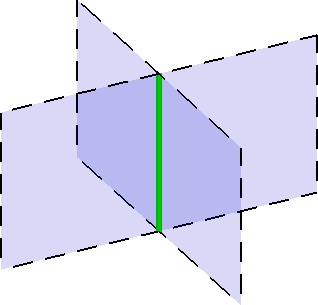}
&	\hspace{2mm} \includegraphics[scale=0.07]{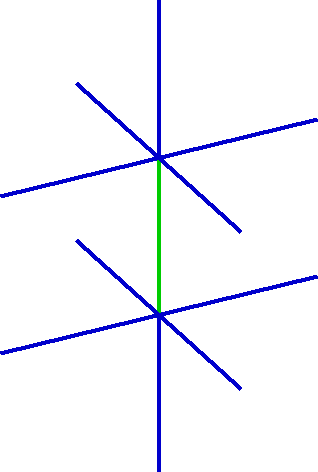}\\

	2
&	\centering \includegraphics[scale=0.07]{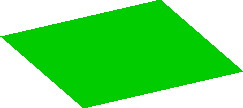}
&	\centering \includegraphics[scale=0.07]{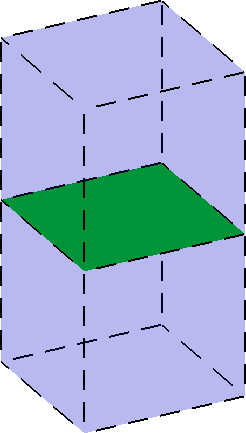}
&	\includegraphics[scale=0.07]{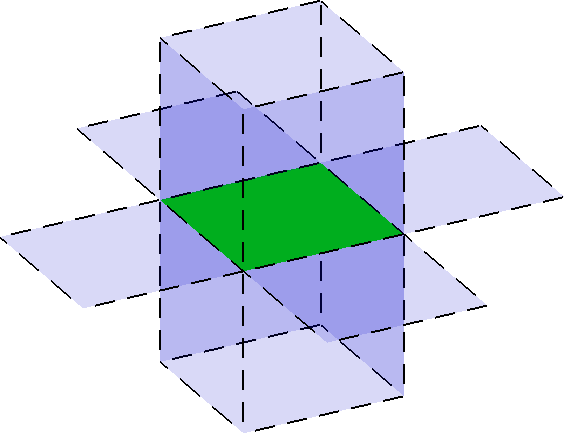}\\

	3
&	\centering \includegraphics[scale=0.07]{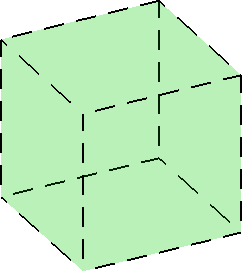}
&	\centering $\emptyset$
&	\includegraphics[scale=0.07]{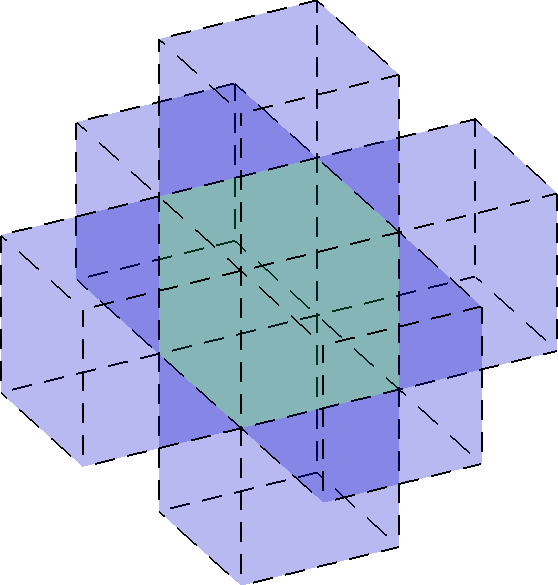}\\
\hline
\end{tabular}
\caption{Two relations are crucial to the definition of connected
components of cells. (i) The $\Gamma$-neighborhood of a $j$-cell
consists of all its 6-neighbors on the topological grid $T$ that 
are $(j+1)$-cells. (ii) The connectivity relation ``$\leftrightarrow$''
connects two cells $t_1, t_2 \in T$ if and only if there exists a 
third cell $t \in T$ such that both $t_1$ and $t_2$ are 
$\Gamma$-neighbors of $t$.}
\label{figure:geo-cell_complex}
\end{SCfigure}

Segments, faces, curves and points can now be defined recursively as connected
components of 3-cells, 2-cells, 1-cells, and 0-cells which are called 
$j$-components. In the following definition, a distinction is made between 
active and inactive cells.

\begin{definition}
\label{definition:connected-components}
Any 3-cell is said to be \emph{active}. A set of all (active) 3-cells that 
belong to the same segment is called a \emph{3-component}.

For $j \in \{2,1,0\}$ and any $j$-cell $t \in T$, let $\{t_1,\ldots,t_{6-2j}\} 
= \Gamma(t)$ be its $\Gamma$-neighbors. For each $k \in \{1, \ldots, 6-2j\}$, 
let $\tau_k$ be the connected component of $t_k$ if $t_k$ is active, and let 
$\tau_k = \emptyset$ otherwise. Define $\theta(t)$ to be the set of connected 
components that occur precisely once in $(\tau_1, \ldots, \tau_{6-2j})$. These
connected components are said to be \emph{bounded} by $t$. Moreover, $t$ is 
called \emph{active} if $\theta(t) \not= \emptyset$.

For each $j \in \{0,1,2\}$, a \emph{$j$-component} is a maximal set 
$U \subseteq T$ with the following properties:

(i) Any $t \in U$ is an active $j$-cell.

(ii) All $t \in U$ bound the same connected components of $(j+1)$-cells,
i.e.~there exists a set $\Theta$ such that $\theta(t) = \Theta$,
for all $t \in U$.

(iii) For any $t_1, t_2 \in U$, there exists a path in $U$ from $t_1$ to $t_2$ 
in which each cell is connected via $\leftrightarrow$ to its predecessor.
\end{definition}

This definition captures not only the geometry but also the topology a a volume 
segmentation. Given, for instance, a face between two segments (i.e.~a 
2-component) and any of its 2-cells, $t$, $\theta(t)$ identifies the two 
segments (3-components) that are bounded by the face. In practice, $\theta(t)$ 
can be stored for each $j$-component. In theory, this corresponds to a cellular
complex representation that is isomorphic to the topology of the volume 
segmentation \cite{kovalevsky-1989}. Cellular complexes are defined as 
follows.

\begin{definition}
A \emph{cellular complex} is a triple $(C, <, \dim)$ in which ``$<$'' is a 
strict partial order in $C$ and $\dim: C \rightarrow \mathbb{N}_0$ maps 
elements of $C$ to non-negative integers such that $\forall c,c' \in C: 
c < c' \Rightarrow \dim(c) < \dim(c')$. The elements of $C$ are called 
\emph{cells}, ``$<$'' the \emph{bounding relation} and $\dim$ the 
\emph{dimension function} of the cellular complex.
\end{definition}

As an example, consider the cellular complex that contains as cells all points
of the topological grid $T$ (these points have already been referred to as 
cells), and as a bounding relation the transitive closure of the 
$\Gamma$-neighborhood, i.e.~the strict partial order that relates any
$t_1, t_2 \in T$ precisely if there exist an $n \in \mathbb{N}$ and a sequence
of $n$ cells $p: \{1, \ldots, n\} \rightarrow T$ such that $p(1) = t_1$, 
$p(n) = t_2$, and $\forall j \in \{2, \ldots, n\}: p(j) \in \Gamma(p(j-1))$.
The dimension function simply maps each cell to its order, either 0, 1, 2 or 3.
This cellular complex corresponds to the topology of the finest possible 
segmentation in which each voxel is a separate segment.

A coarser cellular complex contains as cells the $j$-components 
(Def.~\ref{definition:connected-components}). Its bounding relation is the 
transitive closure of the bounding relation $\theta$ of 
Def.~\ref{definition:connected-components}. Its dimension function maps each 
connected component to the order of its cells. This cellular complex captures 
the topology of the volume segmentation. It would makes sense to refer to its 
elements again as cells. However, to avoid confusion, the termed 
$j$-components is used throughout this article.

An important property of Def.~\ref{definition:connected-components} is that it 
is constructive and thus motivates an algorithm for the labeling of 
$j$-components.

\section{Extraction of Segmentation Geometry and Topology}
\label{section:geo-algorithm}

The geometry of a volume segmentation is made explicit by labeling not only the
segments but also the faces between segments, the curves between faces and the 
points between curves, i.e.~the $j$-components of the segmentation, on the 
topological grid $T$. The resulting \emph{topological label map}
\begin{equation}
\tau: T \rightarrow \mathbb{N}_0
\end{equation}
assigns a positive integer, representative of a $j$-component, to each active 
cell, and zero to all inactive cells. On the computer, $\tau$ is stored as a 
3-dimensional array. 

The first step towards this labeling is to copy all segment labels from the 
volume labeling $\sigma$ to the topological grid labeling $\tau$ by means of 
Algorithm \ref{algorithm:3-cell-labeling}. Subsequently, 2-components and 
1-components are identified and labeled by means of Algorithm 
\ref{algorithm:21-cell-labeling} that performs a depth-first-search. The 
auxiliary function \emph{once} used in this algorithm takes an input 
sequence of integers and returns an ordered sequence of the same length that 
contains those positive integers that occur precisely once in the input 
sequence. Additional entries in the output sequence are filled with zeros. 
Finally, active 0-cells are identified and labeled by means of Algorithm 
\ref{algorithm:0-cell-labeling}. Besides labeling the topological grid, 
these algorithms construct the bounding relation $\theta$ of $j$-components 
(Def.~\ref{definition:connected-components}) and thus a cell complex 
representation of the topology of the segmentation. 

Overall, this connected component labeling is an exact implementation of 
Def.~\ref{definition:connected-components} and is thus known to be correct. Its
runtime is linear in the number of voxels and so is its space complexity. The 
memory dynamically allocated for the stack is in addition bounded by the 
number of cells in the largest face. The absolute memory requirement 
nevertheless renders the procedure impractical for large volume segmentation. 
As shown in the introduction, the topological label map $\tau$ of a volume 
segmentation that consists of 2,000$^3$ voxels is too large to fit in the RAM 
of a desktop computer. Storing $\tau$ on the hard drive and loading blocks into
RAM on demand as a sub-routine of Algorithm \ref{algorithm:21-cell-labeling} 
does not solve the problem because any caching of blocks becomes inefficient if
segments and faces extend unsystematically across large parts of the volume 
which is often the case, in particular in connectomics datasets. Fortunately, 
the labeling itself can be constrained to small blocks of the volume which can 
be chosen systematically and processed independently, with very limited RAM and 
in parallel.

\begin{algorithm}
\label{algorithm:3-cell-labeling}
\caption{Labeling of 3-cells}
\KwIn{$\sigma: G \rightarrow \mathbb{N}$ (segment label map)}
\KwOut{$\tau: T \rightarrow \mathbb{N}$ (topological label map, preliminary),
$n \in \mathbb{N}$ (maximum segment label)}
$n \leftarrow 0$\;
\ForEach{$r \in G$}{
	$\tau(2 r - 1) \leftarrow \sigma(r)$\;
	\If{$\sigma(r) > n$}{
		$n \leftarrow \sigma(r)$\;
	}
}
\end{algorithm}

\begin{algorithm}
\label{algorithm:21-cell-labeling}
\caption{Labeling of 2-cells and 1-cells}
\KwIn{$\tau: T \rightarrow \mathbb{N}$ (topological label map, preliminary), $c \in \{1,2\}$ (cell order)}
\KwOut{$\tau$ (modified), $n \in \mathbb{N}$ (number of $c$-components),
$\alpha \in \mathbb{N}^{n \times (6-2c)}$ (neighborhood relation of $c$-cells)}
$n \leftarrow 0$\;
Stack $s \leftarrow \emptyset$\;
\ForEach{c-cell $t \in T$}{
	\If{$\tau(t) = 0$}{
		$p \leftarrow (6-2c)$\;
		$(t_1, \ldots, t_p) \leftarrow \Gamma(t)$\;
		$(x_1, \ldots, x_p) \leftarrow (\tau(t_1), \ldots, \tau(t_p))$\;
		$(y_1, \ldots, y_p) \leftarrow\ $once$(x_1, \ldots, x_p)$\;
		\If{$y_1 \not= 0$}{
			$n \leftarrow n+1$\;
			\For{$j=1$ \KwTo $p$}{
				$\alpha(n,j) \leftarrow y_j$\;
			}
			$s$.push($t$)\;
			\While{$s \not= \emptyset$}{
				$u \leftarrow s$.pop()\;
				$\tau(u) \leftarrow n$\;
				\ForEach{$v \sim u$}{
					\If{$\tau(v) = 0$}{
						$(v_1, \ldots, v_p) \leftarrow \Gamma(v)$\;
						$(x'_1, \ldots, x'_p) \leftarrow (\tau(v_1), \ldots, \tau(v_p))$\;
						$(y'_1, \ldots, y'_p) \leftarrow\ $once$(x'_1, \ldots, x'_p)$\;
						\If{$(y'_1, \ldots, y'_p) = (y_1, \ldots, y_p)$}{
							$s$.push($v$)\;
						}
					}
				}
			}
		}
	}
}
\end{algorithm}

\begin{algorithm}
\label{algorithm:0-cell-labeling}
\caption{Labeling of active 0-cells}
\KwIn{$\tau: T \rightarrow \mathbb{N}$ (topological label map, preliminary)}
\KwOut{$\tau$ (modified), $n \in \mathbb{N}$ (number of active 0-cells),
$\alpha \in \mathbb{N}^{n \times 6}$ (neighborhood relation of 0-cells)}
$n \leftarrow 0$\;
\ForEach{0-cell $t \in T$}{
	$(t_1, \ldots, t_6) \leftarrow \Gamma(t)$\;
	$(x_1, \ldots, x_6) \leftarrow (\tau(t_1), \ldots, \tau(t_6))$\;
	$(y_1, \ldots, y_6) \leftarrow\ $once$(x_1, \ldots, x_6)$\;
	\If{$y_1 \not= 0$}{
		$n \leftarrow n+1$\;
		\For{$j=1$ \KwTo $6$}{
			$\alpha(n,j) \leftarrow y_j$\;
		}
		$\tau(u) \leftarrow n$\;
	}
}
\end{algorithm}

\section{Block-wise Processing of Large Segmentations}
\label{section:geo-block-wise}
In order to extract the geometry and topology from large volume segmentations 
efficiently with limited RAM, the labeling of components is constrained to
sufficiently small blocks of the topological grid. Each block is processed
independently using the algorithms \ref{algorithm:3-cell-labeling},
\ref{algorithm:21-cell-labeling} and \ref{algorithm:0-cell-labeling}. The
independent results are stored on the hard drive and subsequently combined into
a consistent labeling of the entire topological grid. More precisely, the 
procedure works as follows.

\begin{figure}[t]
a) \includegraphics[width=0.45\columnwidth]{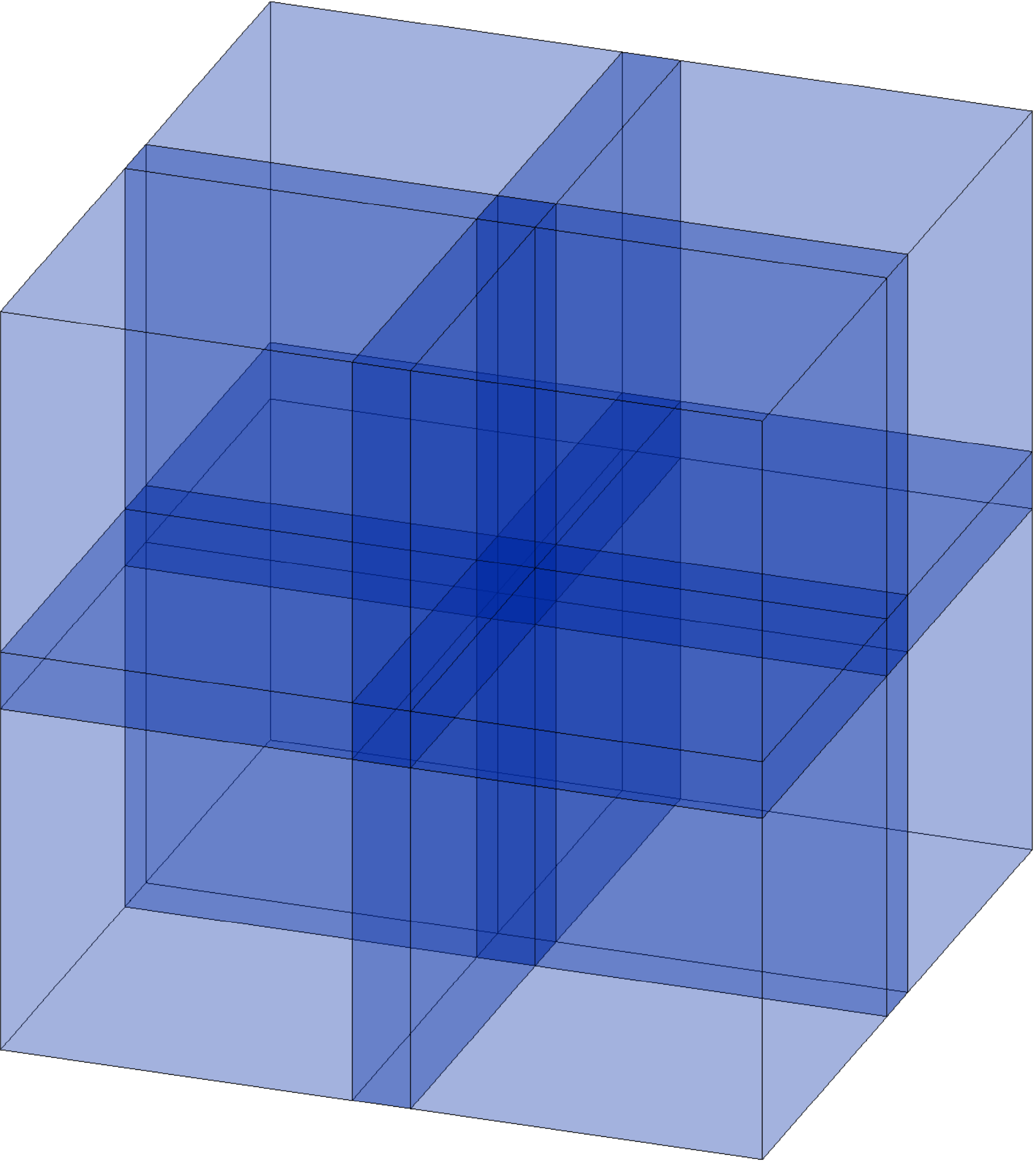}
b) \includegraphics[width=0.45\columnwidth]{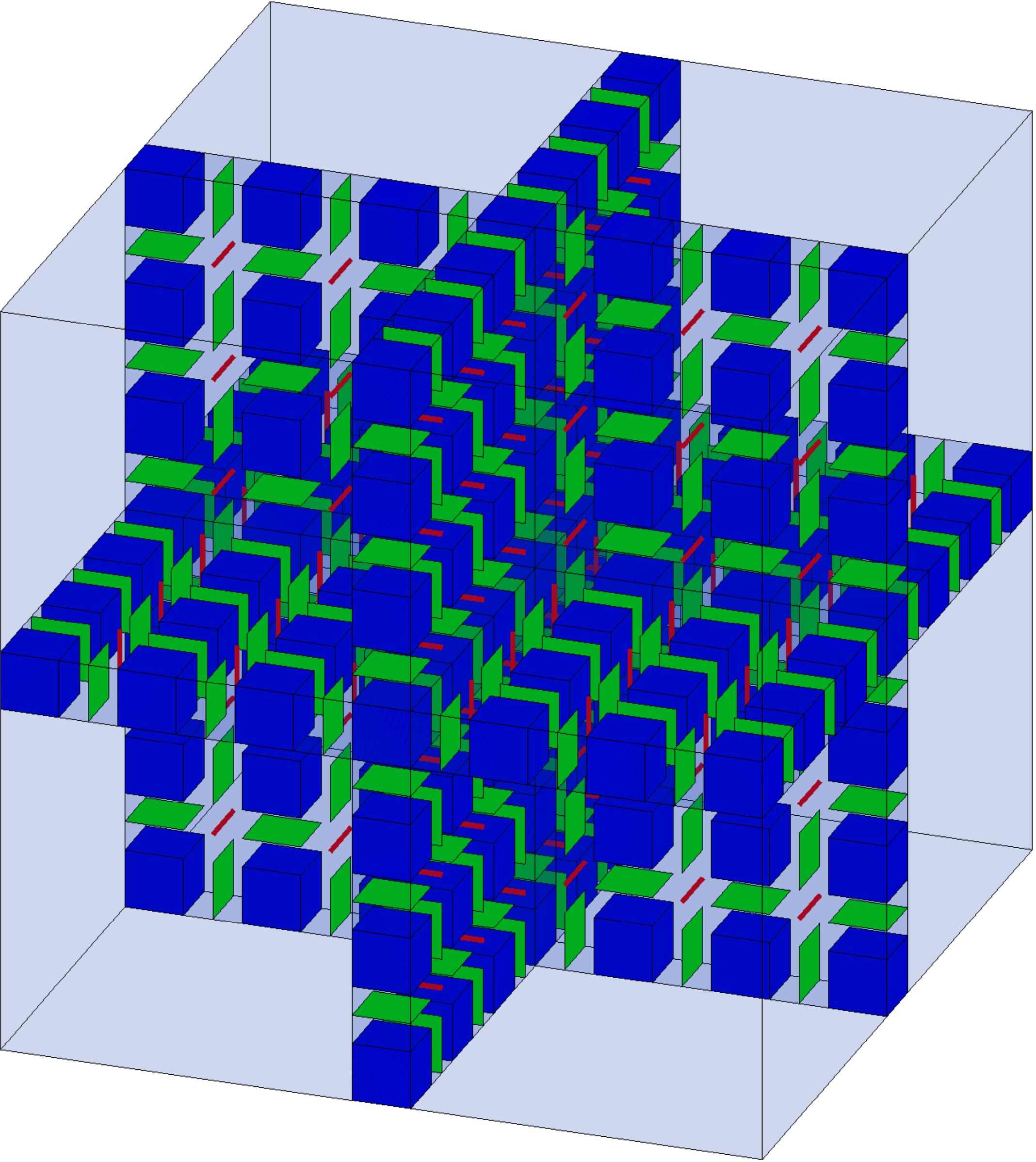}
\caption{The topological grid is subdivided into blocks, leaving an overlap of 
one cell in each direction (a). Cells inside regions of overlap (b) are 
assigned two different labels during the independent processing of the blocks.
These labels are subsequently reconciled.}
\label{figure:geo-blockwise-processing}
\end{figure}

\emph{Step 1 (connected component labeling)}.
The topological grid is subdivided into blocks such that each block begins and 
ends in each direction with a layer that contains 3-cells. Moreover, adjacent 
blocks are chosen to overlap each other by one cell in each direction as is 
depicted in Fig.~\ref{figure:geo-blockwise-processing}. In consequence, each 
1-cell and each 2-cell within a region of overlap belongs to two different
blocks, cf.~Fig.~\ref{figure:geo-blockwise-processing}b. Each block is then 
labeled independently using the algorithms \ref{algorithm:3-cell-labeling}, 
\ref{algorithm:21-cell-labeling}, and \ref{algorithm:0-cell-labeling}, and the 
respective labelings are stored on the hard drive. In consequence, the labeling 
of connected components starts in each block with the label 1.

\emph{Step 2 (label disambiguation)}.
The processed blocks are put in an arbitrary but fixed order. If the $j$-th 
block in this order contains $m_2$ 2-components, the \emph{offset} $m_2$ is 
stored along with block $j+1$ where it is added on demand to all non-zero 
2-cell labels, similarly for 1-cells and 0-cells, arriving at maximal labels 
$M_0, M_1, M_2$ of 0-, 1-, and 2-cells, respectively, for the entire volume.

\emph{Step 3 (label reconciliation)}. 
Whenever connected components of cells extend across block boundaries, their 
labels in the respective blocks (with offsets added) need to be reconciled. Two
disjoint set data structures equipped with the operations \emph{union} and 
\emph{find} \cite{cormen-2001} are used for this purpose, one for 1-cells and 
one for 2-cells, the former with $M_1$, the latter with $M_2$ initially 
distinct sets, each set containing one label. First, union$(l_1, l_2)$ is 
called for the pair $(l_1, l_2)$ of distinct labels assigned to any active 
1-cells and 2-cells within a region of overlap. Second, each label $l$ is 
replaced by the representative find$(l)$ of the union to which it belongs.

\emph{Step 4 (curve merging)}. 
As is elucidated in the correctness analysis of this algorithm in Section 
\ref{section:geo-correctness}, 1-components can still be falsely split and 
0-cells falsely labeled as active at this stage. Thus, in a last step, each 
0-cell $t_0$ and any pair $(t_1, t'_1)$ of 1-components bounded by $t_0$ is 
considered. The labels of $t_1$ and $t'_1$ are reconciled if $t_1$ and $t'_1$ 
bound the same connected components of 2-cells. If at least one reconciliation
has taken place, the activity of $t_0$ is re-computed.
\subsection{Correctness of the Algorithm}
\label{section:geo-correctness}
In order to prove that the block-wise processing is correct, the segment label 
map $\sigma$ as well as the decomposition of the topological grid into blocks 
are assumed to be arbitrary but fixed. The labeling $\tau'$ output by the 
block-wise method is compared to the labeling $\tau$ obtained from the 
application of Algorithms \ref{algorithm:3-cell-labeling}, 
\ref{algorithm:21-cell-labeling} and \ref{algorithm:0-cell-labeling} to the 
entire segment label map. While the latter is known to be correct, the former 
is correct if the two labelings are isomorphic:
\begin{definition}
Two labelings $\tau, \tau': T \rightarrow \mathbb{N}$ of the topological grid 
$T$ are isomorphic w.r.t. a subset $U \subseteq T$ if and only if the 
following conditions hold:
\begin{eqnarray}
& & \forall u \in U: \quad \tau(u) = 0
	\Leftrightarrow \tau'(u) = 0 \enspace , \label{A}\\
& & \forall u, v \in U: \quad \tau(u) = \tau(v)
	\Leftrightarrow \tau'(u) = \tau'(v) \enspace .\label{B}
\end{eqnarray}
If $\tau$ and $\tau'$ are isomorphic w.r.t. the entire domain $T$, they are 
simply called isomorphic.
\end{definition}

\begin{proposition}
$\tau$ and $\tau'$ are isomorphic w.r.t. all 3-cells of $T$.
\end{proposition}

\begin{proof}
3-cell labels are copied from the segment label map $\sigma$ to
$\tau$ and $\tau'$, respectively by both algorithms. The labelings
$\tau$ and $\tau'$ are therefore identical and thus isomorphic
w.r.t. all 3-cells of $T$.
\end{proof}

\begin{proposition}
$\tau$ and $\tau'$ are isomorphic w.r.t. all 2-cells of $T$.
\end{proposition}

\begin{proof}
During block-wise processing, the decision whether or not a 2-cell
obtains a non-zero label depends exclusively on the labeling of
3-cells w.r.t. which $\tau$ and $\tau'$ are identical. Thus, (\ref{A})
holds for all 2-cells of $T$.

Let $u, v \in T$ be any 2-cells.
If $\tau(u) \not= \tau(v)$, $u$ and $v$ bound different pairs of
segments and hence obtain different labels during block-wise
processing. Such labels are not reconciled. Thus,
$\tau'(u) \not= \tau'(v)$.

If $\tau(u) = \tau(v)$, there exists a path of 2-cells between $u$
and $v$ on which all cells separate the same pair of segments. 
If this path is contained in one single block, all its 2-cells obtain 
the same label during the independent processing of that block. 
If the path crosses the boundaries of blocks, the labels along the 
path are reconciled. Thus, $\tau'(u) = \tau'(v)$.

Hence, (\ref{B}) hold for all pairs of 2-cells. In conclusion, $\tau$
and $\tau'$ are isomorphic w.r.t. all 2-cells of $T$.
\end{proof}

\begin{proposition}
\label{proposition:1-cells}
For each block $U \subseteq T$, any 1-cell $t_1 \in U$, and any
2-cell $t_2 \in \Gamma(t_1)$, the label assigned to $t_2$ is unique
among the labels assigned to all 2-cells in $\Gamma(t_1)$ before
label reconciliation if and only if it is unique afterwards.
\end{proposition}

\begin{proof}
($\Rightarrow$) Suppose $t_2$ had a unique label among the elements
of $\Gamma(t_1)$ before label reconciliation and the same label as
another $t'_2 \in \Gamma(t_1)$ afterwards, i.e.~in $\tau'$. Then,
$t_2$ and $t'_2$ separated the same pair of segments because $\tau'$
is isomorphic to the correct labeling $\tau$ w.r.t the 2-cells.
Moreover, we know by definition that $t_2 \leftrightarrow t'_2$,
so $t_2$ and $t'_2$ would have obtained the same label before label
reconciliation. A contradiction. ($\Leftarrow$) Trivial.
\end{proof}

\begin{proposition}
\label{proposition:1-cell-activity}
For all 1-cells $t_1 \in T$ holds
$\tau(t_1) = 0 \Leftrightarrow \tau'(t_1) = 0$.
\end{proposition}

\begin{proof}
During the independent processing of each block, any 1-cell
$t_1 \in T$ obtains a non-zero label if and only if at least one
2-cell label is unique among the labels assigned to all 2-cells of
$\Gamma(t_1)$. In this step, the 2-cell labels before label
reconciliation are considered. However, by Prop.
\ref{proposition:1-cells}, this is no different than considering the
2-cell labels of $\tau'$. Moreover, $\tau'$ and $\tau$ are isomorphic
w.r.t. all 2-cell and thus, 1-cells obtain a non-zero label in $\tau'$
precisely if they are labeled non-zero in $\tau$, i.e.
$\tau(t_1) = 0 \Leftrightarrow \tau'(t_1) = 0$.
\end{proof}
\begin{proposition}
\label{proposition:1-cell-implication}
For all 1-cells $u, v \in T$ holds
$\tau(u) = \tau(v) \Leftarrow \tau'(u) = \tau'(v)$.
\end{proposition}
\begin{proof}
If $\tau'(u) = 0$, the conjecture holds by virtue of 
Prop.~\ref{proposition:1-cell-activity}. If $\tau'(u) \not= 0$, it 
follows that $\tau'(v) \not= 0$ (by assertion) as well as $\tau(u) \not= 0$
and $\tau(v) \not= 0$ (by Prop. \ref{proposition:1-cell-activity}).
Moreover, $\tau'(u) = \tau'(v)$ requires by construction of $\tau'$
that $u$ and $v$ are connected by a path of 1-cells all of which bound
the same connected components of 2-cells (in $\tau'$). As $\tau$ and
$\tau'$ are isomorphic w.r.t.~the 2-cells and as 1-cells are labeled
correctly in $\tau$, it follows that $\tau(u) = \tau(v)$.
\end{proof}
\begin{proposition}
\label{proposition:0-cells}
For all 0-cells $t \in T$ holds
$\tau(t) \not= 0 \Rightarrow \tau'(t) \not= 0$.
\end{proposition}
\begin{proof}
If $\tau(t) \not= 0$, at least one 1-cell label is non-zero and unique
among the labels assigned to all 1-cells in $\Gamma(t)$, i.e.
\[\exists u \in \Gamma(t): \quad \tau(u) \not=0 \wedge \forall v \in \Gamma(t) \setminus \{u\}: \tau(u) \not= \tau(v) \enspace .\]
For any such $u$ follows by Prop.~\ref{proposition:1-cell-activity} and \ref{proposition:1-cell-implication}
\[\tau'(u) \not= 0 \wedge \forall v \in \Gamma(t) \setminus \{u\}: \tau'(u) \not= \tau'(v)\]
and thus, by construction of $\tau'$, the conjecture.
\end{proof}

In order to prove that $\tau$ and $\tau'$ are isomorphic, it
remains to be shown that the inverse implications of 
Prop.~\ref{proposition:1-cell-implication} and \ref{proposition:0-cells}
also hold. Unlike the above propositions which hold by construction of
$\tau'$ in Steps 1--3 of the block algorithm, the two missing
implications are enforced explicitly, by Step 4.

As the following example shows, 1-components can indeed be falsely split and 
0-cells falsely labeled active if Step 4 is omitted. In 
Fig.~\ref{figure:geo-example1}a, a segment label map on a grid of $3 \times 3 
\times 2$ voxels is shown. Six segments are identified by the integers $1$ 
through $6$. The correct corresponding topological label map is depicted in 
Fig.~\ref{figure:geo-example1}b, and the connected components are plotted in 
Fig.~\ref{figure:geo-example2}a and \ref{figure:geo-example2}b. The 1-cell 
labels in Fig.~\ref{figure:geo-example1}b are colored in accordance with the 
graphical visualization in Fig.~\ref{figure:geo-example2}b.

Assume that the segment label map in Fig. \ref{figure:geo-example1}a is 
processed block-wise, with blocks of $2 \times 2 \times 2$ voxels. Note that 
this block-size includes an overlap of one voxel in each direction. The bold 
font in Fig. \ref{figure:geo-example1}a indicates one of these blocks. The 
topological label map that is constructed when this block is processed 
independently is depicted in Fig.~\ref{figure:geo-example1}c. While the 1-cells
labeled 1 and 2 are merged into one connected component during label 
reconciliation after all blocks have been processed, the 1-cells labeled 3 and 
4 are merged only in Step 4 of the algorithm. If Step 4 were omitted, the 
incorrect labeling shown in Fig. \ref{figure:geo-example2}c would be computed.

\begin{SCfigure}[1][t]
\begin{minipage}[t]{7cm}
a) 
\begin{tabular}[t]{|l@{\hspace{2.5mm}}l@{\hspace{2.5mm}}l|}
\hline
\multicolumn{3}{|c|}{z = 1}\\
\hline
\hline
1 & 1 & 1\\
1 & \bf{2} & \bf{1}\\
1 & \bf{1} & \bf{3}\\
\hline
\end{tabular}
\begin{tabular}[t]{|l@{\hspace{2.5mm}}l@{\hspace{2.5mm}}l|}
\hline
\multicolumn{3}{|c|}{z = 2}\\
\hline
\hline
4 & 4 & 4\\
4 & \bf{5} & \bf{4}\\
4 & \bf{4} & \bf{6}\\
\hline
\end{tabular}\\
b)
\begin{tabular}[t]{|l@{\hspace{2.5mm}}l@{\hspace{2.5mm}}l@{\hspace{2.5mm}}l@{\hspace{2.5mm}}l|}
\hline
\multicolumn{5}{|c|}{z = 1}\\
\hline
\hline
1 & 0 & 1 & 0 & 1\\
0 & 0 & 1 & 0 & 0\\
1 & 1 & 2 & 1 & 1\\
0 & 0 & 1 & 0 & 2\\
1 & 0 & 1 & 2 & 3\\
\hline
\end{tabular}
\begin{tabular}[t]{|l@{\hspace{2.5mm}}l@{\hspace{2.5mm}}l@{\hspace{2.5mm}}l@{\hspace{2.5mm}}l|}
\hline
\multicolumn{5}{|c|}{z = 2}\\
\hline
\hline
3 & 0 & 3 & 0 & 3\\
0 & 0 & \bluf{1} & 0 & 0\\
3 & \bluf{1} & 4 & \bluf{1} & 3\\
0 & 0 & \bluf{1} & 0 & \gref{2}\\
3 & 0 & 3 & \gref{2} & 5\\
\hline
\end{tabular}
\begin{tabular}[t]{|l@{\hspace{2.5mm}}l@{\hspace{2.5mm}}l@{\hspace{2.5mm}}l@{\hspace{2.5mm}}l|}
\hline
\multicolumn{5}{|c|}{z = 3}\\
\hline
\hline
4 & 0 & 4 & 0 & 4\\
0 & 0 & 6 & 0 & 0\\
4 & 6 & 5 & 6 & 4\\
0 & 0 & 6 & 0 & 7\\
4 & 0 & 4 & 7 & 6\\
\hline
\end{tabular}\\
c)
\begin{tabular}[t]{|l@{\hspace{2.5mm}}l@{\hspace{2.5mm}}l|}
\hline
\multicolumn{3}{|c|}{z = 1}\\
\hline
\hline
2 & 1 & 1\\
1 & 0 & 2\\
1 & 2 & 3\\
\hline
\end{tabular}
\begin{tabular}[t]{|l@{\hspace{2.5mm}}l@{\hspace{2.5mm}}l|}
\hline
\multicolumn{3}{|c|}{z = 2}\\
\hline
\hline
3 & \bluf{2} & 5\\
\bluf{1} & \bf{1} & \cyf{4}\\
4 & \magf{3} & 6\\
\hline
\end{tabular}
\begin{tabular}[t]{|l@{\hspace{2.5mm}}l@{\hspace{2.5mm}}l|}
\hline
\multicolumn{3}{|c|}{z = 3}\\
\hline
\hline
5 & 7 & 4\\
7 & 0 & 8\\
4 & 8 & 6\\
\hline
\end{tabular}
\end{minipage}
\caption{a) A segment label map on a grid of
$3 \times 3 \times 2$ voxels. b) The correct corresponding topological label
map. Colors are in accordance with the 1-cells shown in
Fig.~\ref{figure:geo-example2}b. c) The topological label map of the block depicted
in bold font in (a). 1-cells are colored in accordance with 
Fig.~\ref{figure:geo-example2}c. The 1-cells labeled 1 and 2 are merged during
label reconciliation while the 1-cells labeled 3 and 4 are merged in Step 4 of 
the block-wise processing.}
\label{figure:geo-example1}
\end{SCfigure}
\begin{figure*}
a) \hspace{0.27\textwidth} b) \hspace{0.3\textwidth} c)\\
\includegraphics[width=0.32\textwidth]{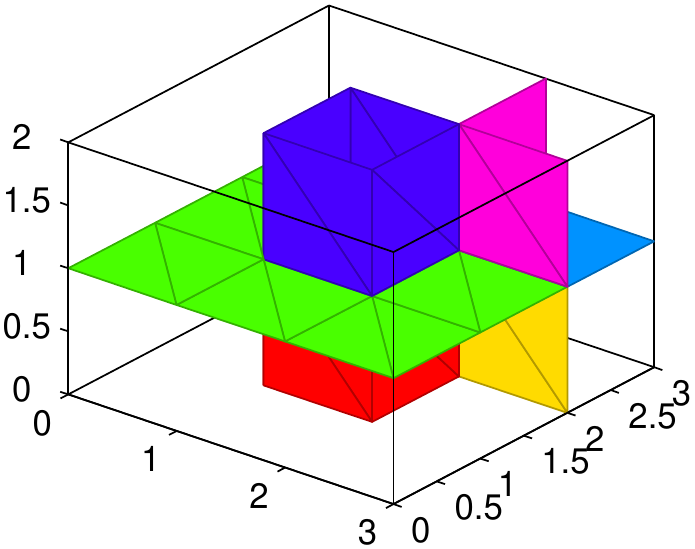}
\includegraphics[width=0.32\textwidth]{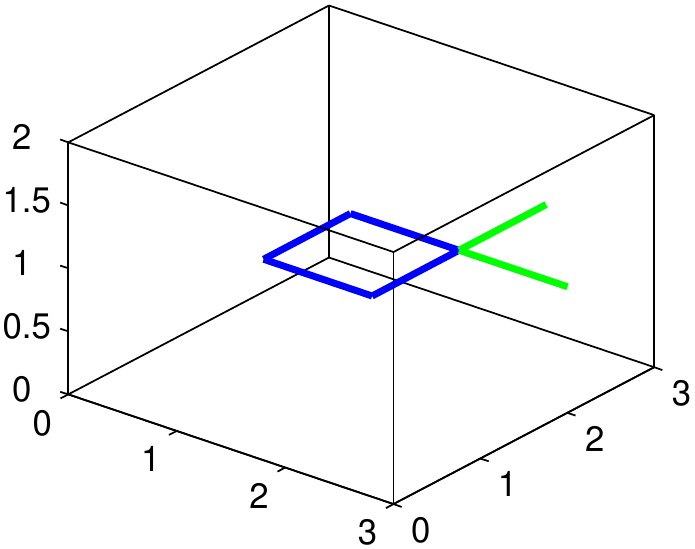}
\includegraphics[width=0.32\textwidth]{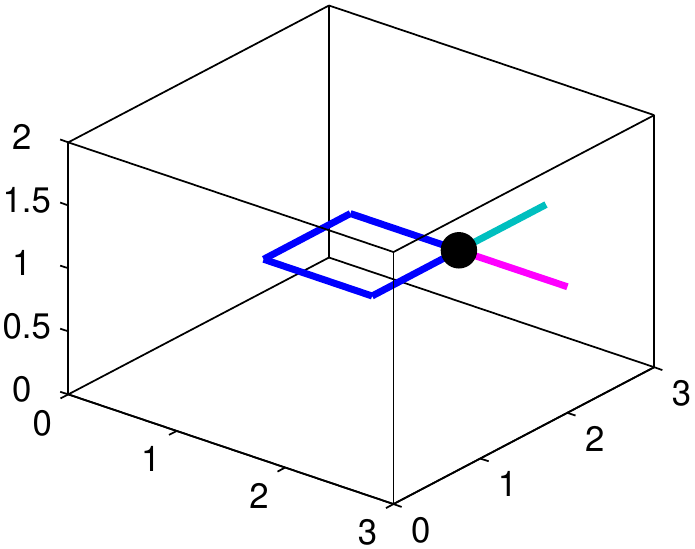}
\caption{Connected components of 2-cells (a) and 1-cells (b) defined
by the topological label map in Fig. \ref{figure:geo-example1}b. c)
1-cells and the 0-cell from a block-wise labeling where Step 4 of the
algorithm is omitted.}
\label{figure:geo-example2}
\end{figure*}
\begin{theorem}
$\tau$ and $\tau'$ are isomorphic.
\end{theorem}
\begin{proof}
In addition to the implications proven above, Step 4 of the block-wise
processing enforces:

(1) For all 1-cells $u, v \in T$ holds
$\tau(u) = \tau(v) \Leftarrow \tau'(u) = \tau'(v)$.

(2) For all 0-cells $t \in T$ holds $\tau(t) \not= 0 \Leftarrow \tau'(t) \not= 0$.
\end{proof}
By virtue of this theorem, the block-wise processing is correct.
\subsection{Complexity}
The runtime overhead introduced by the merging of labels is 
$O((N + M_1 + M_2) \log (M_1 + M_2))$ where $N$ is the number of cells within 
regions of overlap. Time $O(N \log (M_1 + M_2))$ is used for the $O(N)$ calls 
of \emph{union}, whereas time $O((M_1 + M_2) \log (M_1 + M_2))$ is required for
the $M_1 + M_2$ \emph{find}-operations. In practice, this overhead is 
negligible compared to the runtime of the connected component labeling.
\subsection{Parallelization}
The algorithm can be used in two different settings: Blocks can be
processed consecutively in order to extract the geometry of a large
volume segmentation in limited RAM. Perhaps more interestingly,
blocks can be processed in parallel, possibly on several machines,
with virtually no process synchronization or inter-process
communication.

Indeed, if it were not for parallelization, the block-wise connected
component labeling could have been implemented simpler, starting with
only the 2-cells, followed by the disambiguation and reconciliation of
their labels across all blocks even before any 1-cell or 0-cell is
labeled. This would render Step 4 of the block-wise processing unneccessary.
However, the program would have to wait until the 2-cells of \emph{all}
blocks have been labeled before it could label the first 1-cell. In
contrast, the proposed algorithm starts labeling the 1-cells within a
block as soon as the labeling of 2-cells within that block is
finished, regardless of the progress on other blocks.
\section{Redundant Storage for Constant Time Access}
\label{section:geo-redundant-storage}

The algorithm proposed in the last section labels segments, faces between
segments, the curves between faces and the points between curves on
the toplogical grid. The output is a topological label map that provides
constant time access to the label at any topological coordinate. It allows to 
determine in constant time whether there is a face, curve, or point at a given 
location and if so, to determine its label. This is important, e.g.~for the 
visualization of 2-dimensional slices of a segmentation that show not only
segments but also faces, curves, and points. The algorithm makes explicit
the bounding relations between the geometric objects.

However, not only the labels of individual cells and the adjacency of geometric
objects are important but also the set of all cells that belong to the same
component. For image analysis, it can for instance be useful to compute the mean
gray value over a face between two segments. Yet, a list of all 2-cells of the
face cannot be obtained in constant time from the topological grid labeling. 
Thus, a redundant representation of geometry is constructed that contains for 
each $j$-component a list of all its $j$-cells.

This redundant representation as well as the topological grid labeling are 
stored on the hard drive using the Hierarchical Data Format (HDF5). HDF5 was
originally developed by the National Center for Supercomputing Applications 
(NCSA) and is now maintained by the non-profit HDF5-Group\footnote{\url{www.hdf5group.org}}.
It is widely used, especially in the life sciences 
\cite{dougherty-2009}. An HDF5 file contains two principal types of objects, 
groups and datasets. Datasets represent the actual storage containers and are 
multi-dimensional arrays of a unique type, while groups represent an 
organizational concept analogous to a directory that enables the user to 
hierarchically structure the data within the file. Furthermore, attributes may 
be assigned to any dataset or group and contain meta information pertaining to 
the data stored within these objects.

Two HDF5 files are used here. The first file is associated with the labeling 
algorithm of Section \ref{section:geo-block-wise}.
Its structure is depicted in Fig.~\ref{figure:geo-hdf5}a. For each block
and its index $j$ in the order of blocks, a sub-group named $j$ is created in
the group \emph{blocks}. The sub-group $j$ contains the dataset \emph{topological-grid},
a 3-dimensional array that stores the topological label map of the block.
Furthermore, it contains datasets for the neighborhood relations
of connected components as well as the label offsets of the block which are
computed during label disambiguation. During label reconciliation, the datasets
\emph{relabeling-}$k$ and \emph{neighborhood-}$k$ are created in the main file
to store the labeling and the neighborhood relations of $k$-components of the 
entire topological grid.

\begin{SCfigure}
\begin{minipage}[t]{7.7cm}
a)\hspace{3.5cm}b)\vspace{-2mm}\\
\begin{tabular}[t]{llll}
\hline
$t$ & $d$ & Name\\
\hline
D & 1 & segmentation-shape\\
D & 1 & block-shape\\
G &   & blocks\\
G &   & \hspace{0.3cm}$\langle b \rangle$\\
D & 3 & \hspace{0.6cm}topological-grid\\
D & 1 & \hspace{0.6cm}max-labels\\
D & 1 & \hspace{0.6cm}label-offsets\\
D & 2 & \hspace{0.6cm}neighborhood-0\\
D & 2 & \hspace{0.6cm}neighborhood-1\\
D & 2 & \hspace{0.6cm}neighborhood-2\\
D & 1 & max-labels\\
D & 1 & relabeling-1\\
D & 1 & relabeling-2\\
D & 2 & neighborhood-0\\
D & 2 & neighborhood-1\\
D & 2 & neighborhood-2\\
\hline
\end{tabular}
\begin{tabular}[t]{llll}
\hline
$t$ & $d$ & Name\\
\hline
A & 1 & number-of-bins\\
D & 1 & segmentation-shape\\
D & 1 & max-labels\\
D & 2 & 0-cells\\
G &   & 1-components\\
G &   & \hspace{0.3cm}bin-$\langle b \rangle$\\
D & 2 & \hspace{0.6cm}$\langle q \rangle$-$\langle p \rangle$\\
G &   & 2-components\\
G &   & \hspace{0.3cm}bin-$\langle b \rangle$\\
D & 2 & \hspace{0.6cm}$\langle q \rangle$-$\langle p \rangle$\\
G &   & 3-components\\
G &   & \hspace{0.3cm}bin-$\langle b \rangle$\\
D & 2 & \hspace{0.6cm}$\langle q \rangle$-$\langle p \rangle$\\
D & 1 & parts-counters-1\\
D & 1 & parts-counters-2\\
D & 1 & parts-counters-3\\
\hline
\end{tabular}
\end{minipage}
\caption{Two HDF5 files store the information extracted during 
block-wise processing. Along with each data item, its type $t$ (a 
group $G$, a dataset $D$, or an attribute $A$) and dimension $d$ are 
shown. a) The 1st file provides constant-time access to the label of 
any cell as well as to the neighborhoods of connected components of 
cells. b) The 2nd file provides constant-time access to the coordinate
lists of entire segments, faces and curves.}
\label{figure:geo-hdf5}
\end{SCfigure}

Using this file, constant time access to the label of a given cell works as
follows. After identifying a block $k$ to which the $j$-cell of interest
belongs, the label $l$ is read off from the dataset \emph{topological-grid} of
that block. Except for 3-cells and inactive cells, the offset $m$ associated
with cell order $j$ and block $k$ is loaded and the dataset \emph{relabeling-}$j$
accessed at location $l+m$ for the globally consistent label of the cell. In
practice, all data except the topological label map can be kept in RAM.

The second HDF5 file stores one coordinate list for each 1-, 2-, and 3-component
as well as the coordinates of all active 0-cells. Initially, the most straight 
forward group hierarchy was chosen for this file: Three groups associated with 
segments, faces, and curves, each containing one extendible dataset for each
component. In addition, the coordinates of all 0-cells were stored in one 
2-dimensional array. This group hierarchy turned out to be problematic when more 
than 10$^6$ datasets were created per group, using version 1.8.4 of the HDF5
library. To overcome this problem, the more complex group hierarchy shown in 
Fig.~\ref{figure:geo-hdf5}b is used instead. In each of the groups 
\emph{1-components}, \emph{2-components}, and \emph{3-components}, a fixed 
number of sub-groups is created into which datasets containing coordinates lists 
are distributed. Also for performance reasons, the use of extendible datasets was
dropped, which means that due to the block-wise nature of the algorithm, one 
complete $j$-component may be associated with several datasets representing its 
fragments from different blocks. The HDF5 file contains the datasets 
\emph{parts-counters-}$j$ for the number of datasets a single $j$-connected is 
split into.

\subsection{Alternatives}

A more obvious way to store the coordinate lists is to create one binary file
for each list. The use of many files offers the advantage that the coordinate
lists may easily be extended by appending to files, an important asset for
block-wise processing. However, bearing in mind that a segmentation of 2,000$^3$
voxels easily contains in excess of $10^6$ connected components, the vast
amount of files places a heavy burden on the file system, making simple
operations such as copying the data extremely time consuming. In contrast, HDF5
was designed to organize large numbers of binary datasets.

A good alternative to HDF5 is a relational database. Experiments with a 
PostgreSQL database showed promising performance. This approach was nevertheless
abandoned in favor of HDF5 because the necessity to install and configure a 
database might deter potential users from trying out the software.

\section{Conclusion}
\label{section:geo-conclusion}

In this article, a new algorithm for geometry and topology extraction
from large volume segmentations is proposed. In contrast to previous
methods, this algorithms processes volume segmentations in a
block-wise fashion. This facilitates geometry and topology extraction
from large volume segmentations with limited RAM and in parallel. The
geometry is stored in HDF5 files that provide constant time access to 
the labels of segments, faces between segments, curves between faces
and points between curves at any location as well as to lists of 
coordinates that constitute these objects. This representation makes a 
geometric analysis of large volume segmentations practical.

\appendix
\section{Compiling and Installing the Software}
The CGP software is provided as C++ source code with a CMake 2.6 build system
for the command line tools \emph{cgpx} and \emph{cgpr} as well as for MATLAB mex
files. CGP depends on the HDF5 Library (version 1.8.4 or higher\footnote{\url{http://www.hdfgroup.org/HDF5}} and can optionally make use of the 
Message Passing Interface\footnote{\url{http://www.mcs.anl.gov/research/projects/mpi}} (MPI), 
and the Visualization Toolkit\footnote{\url{http://vtk.org}} vtk, . An example segmentation of $50 \times 50 \times 50$ voxels is included, 
along with the according outputs of \emph{cgpx} and \emph{cgpr}. The following 
paragraphs describe how CGP can be compiled on a system that has HDF5 installed.
\subsection{Linux/UNIX and GNU C++}
Unpack the source archive and create a build directory. Execute CMake in this
directory, providing the path to the source as the last parameter:
\begin{small}
\begin{verbatim}
unzip cgp.zip
mkdir build-cgp && cd build-cgp
CMake ../cgp
make
\end{verbatim}
\end{small}
If HDF5, MATLAB, or vtk are installed in a non-standard way, CMake will not find
them automatically. In this case, paths to include files and libraries need to
be set manually in the above call, e.g.~for HDF5:
\begin{small}
\begin{verbatim}
CMake -DHDF5_INCLUDE_DIR=$HOME/inc \
  -DHDF5_LIBRARY=$HOME/lib/libhdf5.so \
  ../cgp
\end{verbatim}
\end{small}
\subsection{Microsoft Windows and VisualStudio}
Unpack the source archive, create a build directory, and use the CMake GUI to
configure. If CMake does not find HDF5, MATLAB, or vtk although they are
installed, set the include paths and library paths for these packages manually. 
CMake will generate a VisualStudio solution file. Open this file, build the 
target ALL\_BUILD in release mode, and install the binaries by building 
the target INSTALL.
\section{Using the Software}
\subsection{From the Command Line}
\label{section:geo-command-line}
The command line tools \emph{cgpx} and \emph{cgpr} compute a representation of
the geometry and topology of a volume segmentation. The segmentation needs to be
stored as a 3-dimensional array of 32-bit unsigned integers in one dataset in
the root group of an HDF5 file. The command line tools are then used as follows:

cgpx $\langle input\ file \rangle$
$\langle dataset \rangle$
$\langle b1 \rangle$
$\langle b2 \rangle$
$\langle b3 \rangle$
$\langle output\ file\rangle$

cgpr $\langle input\ file \rangle\ \langle output\ file \rangle$.

The first tool constructs a labeled topological grid in a block-wise fashion
using the block shape $b1 \times b2 \times b3$. The second tool writes a list of
topological coordinates for each geometric object. Suppose, as an example, that
a segmentation of 2,000$^3$ voxels is stored as a 3-dimensional array in the
dataset \emph{seg} of the HDF5 file \emph{segmentation.h5}. On a desktop
computer equipped with 2~GB of RAM, a block-size of $200^3$ voxels is
reasonable, so a representation of the geometry and topology of the segmentation
can be obtained like this:

\begin{small}
\begin{verbatim}
cgpx segmentation.h5 seg 200 200 200 grid.h5
cgpr grid.h5 objects.h5
\end{verbatim}
\end{small}
In order to process several blocks in parallel, invoke the command line tool
\emph{cgpx} via \emph{mpiexec}, e.g.

\begin{small}
\begin{verbatim}
mpiexec -n 2 cgpx segmentation.h5 seg 
  200 200 200 grid.h5
\end{verbatim}
\end{small}
\subsection{From MATLAB}
In MATLAB, segmentations are conveniently stored as 3-dimensional arrays whose
entries are 32-bit unsigned integers that correspond to segment labels. In order
to extract the geometry and topology from a segmentation, the array has to be
written to an HDF5 file by means of the function \emph{cgp\_save}. The following
call of \emph{cgp\_save} writes the array $S$ as the dataset \emph{seg} into the
HDF5 file \emph{seg.h5}:

\begin{small}\verb$cgp_save('seg.h5', '/seg', S);$\end{small}\\
Geometry and topology extraction can now be performed either from the command
line, using the tools \emph{cgpx} and \emph{cgpr} as described in Section
\ref{section:geo-command-line}, or directly from MATLAB, using the according
mex-functions:

\begin{small}
\begin{verbatim}
cgpx('seg.h5', 'seg', uint32([b1 b2 b3]),
  'grid.h5');
cgpr('grid.h5', 'objects.h5');
\end{verbatim}
\end{small}
where $b1 \times b2 \times b3$ specifies the block shape.

A number of functions named with the prefix \emph{cgp} can be used to
selectively load data from the geometry file. In the following example, one
curve and its adjacent faces are plotted.
\begin{small}
\begin{verbatim}
desc = cgp_open('objects.h5');
curve_id = 100;
hold on;
tcl = cgp_load_object(desc, 1, curve_id);
cgp_plot_1cells(tcl);
neighbors = desc.neighborhoods{2}(curve_id,:);
for j = 1:length(neighbors)
  if neighbors(j) == 0
      break;
  end
  tcl = cgp_load_object(desc,2,neighbors(j));
  tri = cgp_triangulate(tcl);
  cgp_plot_triangulation(tri, rand(1,3),0.7);
end
hold off;
cgp_close(desc);
\end{verbatim}
\end{small}

A 0-cell and its adjacent curves are plotted as follows:
\begin{small}
\begin{verbatim}
desc = cgp_open('objects.h5');
point_id = 100;
hold on;
tcl = cgp_load_object(desc, 0, point_id);
cgp_plot_0cells(tcl);
neighbors = desc.neighborhoods{1}(point_id,:);
for j = 1:length(neighbors)
  if neighbors(j) == 0
      break;
  end
  tcl = cgp_load_object(desc,1,neighbors(j));
  cgp_plot_1cells(tcl, rand(1,3));
end
hold off;
cgp_close(desc);
\end{verbatim}
\end{small}
\subsection{From C++}
The C++ API for parallelized geometry and topology extraction is defined in the
header files \emph{cgp\_hdf5.hxx}, \emph{CgpxMaster.hxx}, and 
\emph{CgpxWorker.hxx}. The construction of the topological grid is invoked by
classes
\begin{small}
\begin{verbatim}
template<class T, class C> class CgpxMaster;
template<class T, class C> class CgpxWorker;
\end{verbatim}
\end{small}
that implement a master-worker-scheme using MPI. The type $T$ is used for 
labels of geometric objects; unsigned integers of at least 32 bits should be 
used. The type $C$ is used for coordinates to navigate in arrays. 16-bit 
integers are sufficient if the segmentation is smaller than 32769 voxels in 
each dimension.

The coordinate lists of all geometric objects can be constructed from the
topological grid by means of the function
\begin{small}
\begin{verbatim}
template<class T, class C>
void geometry3blockwise(
  const hid_t&, // input HDF5 file
  const hid_t&  // output HDF5 file
);
\end{verbatim}
\end{small}
The source files of the command line tools, 
\begin{small}
\begin{verbatim}
src/cmd/cgpx.cxx
src/cmd/cgpr.cxx
\end{verbatim}
\end{small}
show the interested reader how the classes and function are used.
\end{document}